\documentclass[12pt]{article}
\usepackage[left=3cm,bottom=3cm,right=3cm,top=3cm]{geometry}
\usepackage{amstext}
\usepackage{latexsym}
\usepackage[utf8]{inputenc}
\usepackage[T1]{fontenc}
\usepackage{amsmath, amsthm}
\usepackage{graphicx}
\usepackage{diagbox}
\usepackage{multirow}
\usepackage[margin=10pt,font=small,labelfont=bf,labelsep=endash]{caption}
\usepackage{mathrsfs}
\usepackage{amssymb}
\usepackage{datetime}
\usepackage{tikz}
\usepackage{abstract}
\usepackage{color}

\usetikzlibrary{fit,shapes.misc}

\numberwithin{equation}{section}
\newtheorem{defi}{Definition}[section]
\newtheorem{lemma}{Lemma}[section]

\newtheorem{prop}[lemma]{Proposition}
\newtheorem{theorem}[lemma]{Theorem}
\newtheorem{rem}[lemma]{Remark}
\newtheorem{con}[lemma]{Conjecture}

\begin{document}
\title{Models for Self-Gravitating Photon Shells and Geons}

\author{H{\aa}kan Andr\'{e}asson, David Fajman, Maximilian Thaller}

\date{November 4th, 2015}

\maketitle

\begin{abstract}
We prove existence of spherically symmetric, static, self-gravitating photon shells as solutions to the massless Einstein-Vlasov system. The solutions are highly relativistic in the sense that the ratio $2m(r)/r$ is close to $8/9$, where $m(r)$ is the Hawking mass and $r$ is the area radius. In 1955 Wheeler constructed, by numerical means, so called idealized spherically symmetric geons, i.e. solutions of the Einstein-Maxwell equations for which the energy momentum tensor is spherically symmetric on a time average. The structure of these solutions is such that the electromagnetic field is confined to a thin shell for which the ratio $2m/r$ is close to $8/9$, i.e., the solutions are highly relativistic photon shells. The solutions presented in this work provide an alternative model for photon shells or idealized spherically symmetric geons.

\end{abstract}

\section{Introduction}
In the last 25 years the existence problem for the static Einstein-Vlasov system has been studied intensively. Several works have established the existence of mainly static but also of stationary solutions modeling ensembles of self-gravitating collisionless massive particles in equilibrium. The results include spherically symmetric, asymptotically flat, globally regular solutions with energy density of compact support \cite{rr92_static,rein93,rein_shells,rr00,a072,rare13}, solutions with black holes \cite{rein93}, axially symmetric static~\cite{akr11} and stationary solutions~\cite{akr14}, and static solutions of several types for non-vanishing cosmological constant \cite{aft15}. These works have in common that they consider the case of massive particles and thereby primarily provide models for large scale objects such as galaxies or galaxy clusters.\par
By setting the rest mass in the Einstein-Vlasov system to zero one obtains the corresponding massless system, which describes an ensemble of uncharged self-gravitating massless particles, e.g.~photons. The Cauchy problem for the massless Einstein-Vlasov system has been studied for small spherically symmetric perturbations of Minkowski space by Dafermos \cite{Da06}. Recently numerical evidence for the existence of static solutions to this system was given by Akbarian and Choptuik \cite{AkCh14}. It is the purpose of this work to prove that a class of static solutions to the massless Einstein-Vlasov system indeed exists. \par

\subsection*{Photon shells as particle-like solutions}

The results presented in this paper may also be considered in the wider context of the theory of static solutions to the Einstein equations. For the vacuum Einstein-equations and the Einstein-Maxwell equations classical rigidity theorems imply that the only regular, asymptotically flat, static solution is Minkowski space. Due to this, the construction of models for isolated self-gravitating objects resorts to other types of matter coupled to the Einstein equations. While classical matter models such as fluids, massive collisionless particles, and elastic matter are primarily used for describing macroscopic objects, one motivation for considering field theoretical equations is to construct models for elementary particles - so called \emph{particle-like solutions}. In 1955 Wheeler introduced the concept of \emph{geons} which are singularity free solutions of the Einstein-Maxwell equations, cf.~also \cite{ab97} and \cite{BrHa64}.
His main purpose was to construct solutions which describe effectively a localized particle-like object without introducing mass or charge into the system - the expressions \emph{mass without mass} and \emph{charge without charge} originate from~\cite{W55}. In order to obtain tractable equations Wheeler considered the case of, what he called, \emph{idealized spherically symmetric geons}. These solutions are spherically symmetric only on time average. In Wheeler's treatment the details of the electromagnetic field are thus replaced by the time average of the components of the electromagnetic stress-energy tensor. The resulting equations are solved numerically, cf.~\cite{W55}, and a thin shell of self-gravitating electromagnetic radiation is obtained for which the ratio $2m/r$ is close to $8/9$, where $m$ is the ADM mass and $r$ is the area radius. In fact, Wheeler's numerical solution is closely related to the result we prove in this work. Indeed, if the massless particles are photons, the solutions we construct describe thin photon shells. As in Wheeler's study the ratio $2m/r$ of the shells is close to $8/9$ also in our case. We point out that necessarily $2m/r<8/9$ for any static solution of the Einstein-Vlasov system, cf.~\cite{a08}. This bound also applies in Wheeler's case since the form of the energy momentum tensor for the Maxwell field studied in~\cite{W55} satisfies the assumptions in~\cite{a08}. In conclusion, the massless Einstein-Vlasov system and the Einstein-Maxwell system may in this case be considered as two different mathematical models describing the same physical phenomena, where the former is a particle model and the latter is a field theoretical model. \par
More recently, the existence of a different type of particle-like solutions has been established, i.e., regular spherically symmetric static solutions of the Einstein-Yang-Mills system.
Numerical evidence for their existence was first found by Bartnik and McKinnon \cite{BaMc88}. Eventually, a rigorous proof was obtained by Smoller et al.~\cite{SmWaYaMc91}. It is important to point out that all particle-like solutions discussed above are likely unstable, cf. \cite{W55,AkCh14,ar06,BaMc88}. 

\subsection*{Difficulties and idea of proof}
As mentioned above, for the massive Einstein-Vlasov system several results on existence of static spherically symmetric solutions are available that cover classes of different generality. An essential part in these proofs is to show that the matter quantities are compactly supported. This is clearly a necessary condition for modeling isolated objects. \par
The majority of the existing proofs work indirectly and rely on the asymptotic behavior of the metric, and they have in common that they show the \textit{existence} of a radius that bounds the support of the matter. The value of which however remains undetermined. \par
For the massless Einstein-Vlasov system the situation is different. In section \ref{proofs} we introduce what we call the {\em characteristic function}
\begin{equation} \label{gamma_pre}
\gamma(r) = y(r)-\frac 1 2 \ln\left[m_p^2+\frac{L_0}{r^2}\right],
\end{equation}
where $y$ is defined in \eqref{def_y}, $m_p$ is the particle mass, and $L_0$ is a fixed constant. This function has the property that if and only if $\gamma(r) \leq 0$ all matter quantities vanish at this radius $r$. Equation (\ref{eq_yp}) implies that $y$ is monotonically decreasing in $r$. Consider now the massive case with $m_p = 1$. In this case the logarithm in (\ref{gamma_pre}) is positive for all $r$ and it is known \cite{rare13} that there exists a radius $\tilde R$ such that $y(\tilde R)=0$. Hence $\gamma(r) \leq 0$ for all $r\geq \tilde R$. In the massless case however, when $m_p=0$, the term containing the logarithm behaves completely differently. Assume now that there is a radius $R_1$ such that $\gamma(R_1) = 0$ and $\gamma'(R_1) < 0$. For $r > R_1$ and as long as $\varrho(r)=p(r)=0$, $y$ is given by the Schwarzschild component (\ref{yss}) and we have
\begin{equation}
\gamma(r) = \ln(E_0) -\frac 1 2 \ln\left[1-\frac{2m(R_1)}{r}\right]-\frac 1 2 \ln\left[\frac{L_0}{r^2}\right].
\end{equation}
In view of \cite{a08} we have $\frac{2m(R_1)}{r} \leq \frac 8 9$. So we see that $\gamma(r)$ becomes positive again at some radius and matter reappears. In particular there exists no radius $\tilde R$ such that $\gamma(r)\leq 0$ for all $r \geq \tilde R$. The asymptotic behavior resembles an infinitely extended thin atmosphere of massless Vlasov matter. \par
In this work we use the approach in~\cite{a072} which shows the existence of static spherically symmetric arbitrarily thin shell solutions of the massive Einstein-Vlasov system. More precisely, in~\cite{a072} it is shown that the matter quantities are supported in an interval $[R_0,R_1]$, where the ratio $R_1/R_0$ can be made as close to 1 as desired. The motivation for the study~\cite{a072} was to show that the inequality $2m(r)/r\leq 8/9$ can be arbitrarily well saturated by solutions to the massive Einstein-Vlasov system. In fact, thin shell solutions correspond to highly relativistic solutions for which the ratio $2m(r)/r$ is close to $8/9$. In order to show that the interval of support $[R_0,R_1]$ has the property that $R_1/R_0$ is close to 1 it is of course necessary to obtain detailed control of the radius $R_1$, which is the smallest radius greater than $R_0$ for which the matter quantities vanish. When such a radius has been established the ansatz for the matter distribution can be dropped and the matter quantities are set to zero for $r\geq R_1$. The metric in this -now infinitely extended- vacuum region is the Schwarzschild metric with a mass parameter corresponding to the mass of the enclosed matter shell. In this way a pure matter shell is obtained. We remark that if the solution in the massive case is not continued by a Schwarzschild solution at $r=R_1$, other types of solutions can be constructed, e.g. solutions which contain several shells separated by vacuum regions or separated by thin atmospheres, cf.~\cite{ar07}. These solutions are nevertheless asymptotically flat, i.e., there exists a radius $R_2>R_1$ such that the matter quantities vanish for $r\geq R_2$. Furthermore, in the massive case the ratio $2m/r$ does not have to be close to its upper bound $8/9$ but can be small. Thus static solutions of the massive Einstein-Vlasov system with matter quantities of bounded support do not have to be highly relativistic. \par
In this work we modify the method in~\cite{a072} to the massless case and obtain solutions supported in an interval $[R_0,R_1]$ such that the solutions agree with the Schwarzschild solution for $r\geq R_1$. As was outlined above, it is here necessary to continue the solution with a Schwarzschild solution at a radius where the matter quantities become zero since the solution given globally by an ansatz of the form (\ref{ansf}) is not asymptotically flat.\par

\subsection*{Criteria for shell formation}
A crucial condition for proving the existence of a radius $R_1$, where the matter quantities vanish, is that the there is a radius, called $r_2$, such that the ratio $2m(r_2)/r_2\geq 4/5$. Hence the solutions we construct are highly relativistic.
It is an important question to decide whether or not this is a necessary condition for the existence of static photon shells.
The answer to this question is essential for the physical interpretation of the solution. We conclude the paper by providing numerical evidence for the following conjecture.
\begin{con} \label{our_con}
Let $m(r)$ be the Hawking mass of a solution of the massless Einstein-Vlasov system given in (\ref{def_m}) and $r\in [0,\infty)$ be the areal radius. Then a necessary condition for the existence of massless, asymptotically flat solutions is
\begin{equation}
\Gamma := \sup_{r\in [0,\infty)} \frac{2m(r)}r\geq \frac{4}{5}.
\end{equation}
\end{con}
\noindent In particular, we conjecture that the photon shells are necessarily highly relativistic.

\subsection*{Outline of the paper} This work is organized as follows. In section \ref{preliminaries} we introduce the massless Einstein-Vlasov system in spherical symmetry and review known results that this work relies on. In section \ref{sec_theorem} we state our main result, and section \ref{proofs} is devoted to its proof. Finally, in section \ref{sec_numerics}, we solve the equations numerically to investigate which conditions on the parameters are \textit{necessary} for the existence of photon shells.

\subsection*{Acknowledgments}
We are grateful for discussions with Peter Aichelburg, Bobby Beig, Piotr Chru\'sciel, Walter Simon, Olof Giselsson, and Carsten Gundlach. The second author gratefully acknowledges the hospitality of Chalmers University of Technology, where parts of this work have been carried out.

\section{Preliminaries} \label{preliminaries}

For a given space-time $(M,g)$, the mass-shell of massless particles is defined to be
\begin{equation}
\mathscr P=\{(x,p)\in TM\,:\,g(p,p) =0, p\;\mbox{future directed}, p\neq 0\}.
\end{equation}
In general the particle rest mass $m_p$ is given by the expression $g(p,p) = -m_p^2$. So this relation constitutes the essential difference to the massive system.\par
We consider the static, massless Einstein-Vlasov system in spherical symmetry. First, we introduce the system of equations. In Schwarzschild coordinates the metric is of the form
\begin{equation}
g = -e^{2\mu(r)} \mathrm dt^2 + e^{2\lambda(r)} \mathrm dr^2 + r^2 \mathrm d\vartheta^2 + r^2 \sin^2\vartheta \mathrm d\varphi^2,
\end{equation}
where $r\in [0,\infty)$, $\vartheta\in[0,\pi)$, and $\varphi\in[0,2\pi)$. Throughout this work units are chosen such that $c=G=1$, where $G$ is the gravitational constant and $c$ is the speed of light. The Einstein equations reduce to
\begin{eqnarray}
e^{-2\lambda}\left(2r\lambda'-1\right)+1 &=& 8\pi r^2\varrho(r) \label{eeq1}, \\
e^{-2\lambda}\left(2r\mu'+1\right)-1 &=& 8\pi r^2p(r). \label{eeq2} 
\end{eqnarray}
A prime denotes the derivative with respect to the coordinate $r$. Since we are interested in asymptotically flat solutions we assume the boundary conditions
\begin{equation} \label{bound1}
\lim_{r\to\infty} \mu(r) = \lim_{r\to\infty} \lambda(r) = 0.
\end{equation}
A regular center at $r=0$ is guaranteed by the boundary condition
\begin{equation} \label{bound_center}
\lambda(0)=0.
\end{equation}
The matter considered in this work consists of an ensemble of free falling particles of zero rest mass. These particles are modeled by a density function $f:\mathscr P \to \mathbb R_+$, $(t,x,v)\mapsto f(t,x,v)$, where $t\in\mathbb R$, $x\in\mathbb R^3$, $v\in\mathbb R^3\setminus\{0\}$. In a static setting this function is not depending on $t$ and spherical symmetry means $f(x,v) = f(Ax,Av)$ for any matrix $A\in SO(3)$. The matter distribution function $f$ satisfies the Vlasov equation
\begin{equation} \label{eq_vlasov}
\frac{e^{\mu-\lambda}}{|v|}v^a\frac{\partial f(x,v)}{\partial x^a}-e^{\mu-\lambda}\frac{x^a}{r}|v|\mu'\frac{\partial f(x,v)}{\partial v^a} = 0
\end{equation}
where $|v| = \sqrt{\delta_{ab}v^a v^b}$ and we are using coordinates $v$ such that we have an orthonormal frame on the mass shell $\mathscr P$. See \cite{rr92} for a definition of $v$ and a derivation of the full system in spherical symmetry. The equations are coupled via the matter quantities $\varrho$, $p$, and $p_T$ which are given by
\begin{align}
\varrho(r) &= \int_{\mathbb R^3\setminus\{0\}} f(x,v) |v|\; \mathrm dv^1\mathrm dv^2 \mathrm dv^3, \label{mat11}\\
p(r) &= \int_{\mathbb R^3\setminus\{0\}} \frac{f(x,v)}{|v|} \left[\frac{\delta_{ab} x^a v^b}{r}\right]^2 \mathrm dv^1\mathrm dv^2 \mathrm dv^3, \label{mat12}
\end{align}
and
\begin{equation}
p_T(r) = \frac 1 2 \int_{\mathbb R^3\setminus\{0\}} \left|\frac{x \times v}{r}\right|^2 \frac{f(x,v)}{|v|}\; \mathrm dv^1 \mathrm dv^2 \mathrm dv^3. \label{mat13}
\end{equation}
Consider the quantities  $E$ and $L$, defined by
\begin{equation}
E=e^\mu|v| \quad\mathrm{and}\quad L = x^2v^2-(x \cdot v)^2.
\end{equation}
The quantity $E$ can be seen as particle energy and $L$ is the square of the angular momentum of the particles. Observe that they are conserved along the characteristics of the Vlasov equation (\ref{eq_vlasov}). Moreover if one takes $f$ to be a function $\Phi(E,L)$ of these quantities $E$ and $L$ the Vlasov equation is automatically satisfied. In this paper we choose 
\begin{equation} \label{ansf}
f(x,v) = \Phi(E,L) = \left[1-\frac{E}{E_0}\right]_+^k \left[L-L_0\right]_+^\ell,
\end{equation}
where $\ell \geq -\frac 1 2$, $k \geq 0$, $L_0 > 0$, and $E_0 > 0$. Furthermore $[x]_+ = x$ if $x>0$ and $[x]_+=0$ if $x \leq 0$. Note that in contrast to Newtonian gravity in the relativistic case not all solutions of the Vlasov equation can be written in this form \cite{s99}. \par
For technical reasons we perform a change of variables in the metric coefficient $\mu$. We introduce
\begin{equation} \label{def_y}
y(r) = \ln(E_0) - \mu(r) \qquad \Leftrightarrow \qquad e^\mu = E_0 e^{-y}.
\end{equation}
Note that the Einstein equations (\ref{eeq1}) and (\ref{eeq2}) imply
\begin{equation}\label{eq_yp}
y'(r)=-\frac{1}{1-\frac{2m(r)}{r}}\left[4\pi r p(r) +\frac{m(r)}{r^2}\right]
\end{equation}
where the Hawking mass $m(r)$ is defined to be
\begin{equation} \label{def_m}
m(r) = 4\pi \int_0^r s^2 \varrho(s) \mathrm ds.
\end{equation}
We return to the matter quantities given by (\ref{mat11}) -- (\ref{mat13}). If one inserts the ansatz (\ref{ansf}) for $f$ one obtains by a change of variables in the integrals the expressions
\begin{align}
\varrho(r) &= c_\ell r^{2\ell} \int_{\sqrt{\frac{L_0}{r^2}}}^{e^y} \left(1-\varepsilon e^{-y}\right)^k \varepsilon^2 \left[\varepsilon^2-\frac{L_0}{r^2}\right]^{\ell + \frac 1 2} \mathrm d\varepsilon, \label{mat21} \\
p(r) &= \frac{c_\ell}{2\ell + 3} r^\ell \int_{\sqrt{\frac{L_0}{r^2}}}^{e^y} \left(1-\varepsilon e^{-y}\right)^k \left[\varepsilon^2-\frac{L_0}{r^2}\right]^{\ell + \frac 3 2} \mathrm d\varepsilon \label{mat22}
\end{align}
where
\begin{equation}
c_\ell = 2\pi \int_0^1 \frac{s^\ell}{\sqrt{1-s}}\mathrm ds.
\end{equation}
To have good control of the support of the matter quantities we define for $r\in(0,\infty)$ the \emph{characteristic function}
\begin{equation} \label{defgamma}
\gamma(r) = y(r) -\frac 1 2 \ln\left[\frac{L_0}{r^2}\right].
\end{equation}
Note that the matter quantities vanish if and only if $\gamma(r) \leq 0$ as can be seen by the limits of the integration in (\ref{mat21}) -- (\ref{mat22}). Those limits are in turn a consequence of the cutoff energy $E_0$ and angular momentum $L_0$ in the ansatz function (\ref{ansf}). So $\gamma(r) \leq 0$ is already equivalent to $f(x,v)=0$. We define the radius
\begin{equation} \label{defr0}
R_0 = \sqrt{\frac{L_0}{e^{2y(0)}}}.
\end{equation}
This radius characterizes an inner vacuum region which can be seen as follows. If $L_0>0$ we have $\gamma(r) < 0$ for small $r$ since $y$ is approaching the finite initial value $y_0$ but the logarithm is unbounded. By equation (\ref{eq_yp}) we have $y(R) = y_0$ as long as there is vacuum for all $r\leq R$. So the characteristic function is given by $\gamma(r) = y_0 - \frac 1 2 \ln(L_0/r^2)$ and its first zero will be exactly at $R_0$.\par
In an analogous way as described in \cite{rein93} for the massive Einstein-Vlasov system one can derive the Tolman-Oppenheimer-Volkov equation
\begin{equation} \label{tov_equation}
p'=y'(\varrho+p)-\frac{2}{r}(p-p_T)
\end{equation}
from the expressions (\ref{mat11}) -- (\ref{mat13}) of the matter quantities. It will be used in an integrated form in the proof of Proposition (\ref{prop_m_r}). \par

\section{Main Result} \label{sec_theorem}

The existence of $C^1$-solutions, not necessarily of compact support, of the massless Einstein-Vlasov system follows in an analogous way to the case where massive particles are considered. To make this precise we define {\em quasi shell solutions} whose existence can be taken as starting point for the arguments presented in the following sections.

\begin{defi} \label{def_quasi_shell}
A quasi shell solution of the massless Einstein-Vlasov system is a triple of functions $\lambda,\mu \in C^1([0,\infty))$ and $f\in L^1\left(\mathbb R^6\right)$ such that the Einstein-Vlasov system (\ref{eeq1}) -- (\ref{eq_vlasov}) is fulfilled and the matter quantities (\ref{mat11}) and (\ref{mat12}) are differentiable functions. Furthermore there exists $0 < R_0 < \infty$ such that $f$ and thus also $\varrho$ and $p$ are supported on $(R_0,\infty)$.
\end{defi}

In \cite{rein93} the existence of spherically symmetric, static solutions of the Einstein-Vlasov system for a non vanishing rest mass is proved. The arguments and techniques presented in this work also yield the existence of these quasi shell solutions. It is of course important to note that they do not yield solutions with matter quantities of compact support. In virtue of \cite{rein93} we can thus state the following lemma.

\begin{lemma} \label{lem_quasi}
Suppose the solution $f$ of the Vlasov equation (\ref{eq_vlasov}) is given by an ansatz of the form (\ref{ansf}) where $L_0 > 0$. Then for every initial value $\mu_0 \in \mathbb R$ of the lapse function there exists a unique quasi shell solution of the massless Einstein-Vlasov system (\ref{eeq1}) -- (\ref{eq_vlasov}).
\end{lemma}

The following theorem is the main result of this work and establishes the existence of spherically symmetric solutions of the massless Einstein-Vlasov system with compactly supported matter quantities.

\begin{theorem} \label{new_main_theorem}
There exist static, spherically symmetric, asymptotically flat solutions to the massless Einstein-Vlasov system, with compactly supported matter quantities. These solutions have the property that 
\begin{equation}\label{bounds-2}
\frac{4}{5}<\sup_{r\in[0,\infty)}\frac{2m(r)}{r}<\frac{8}{9}.
\end{equation}
\end{theorem}

\begin{rem}
The upper bound in \eqref{bounds-2} results from the Buchdahl inequality \cite{a08}. The lower bound in \eqref{bounds-2} is investigated numerically in section \ref{sec_numerics} and appears to be a necessary criterion for existence.
\end{rem}
The major step in the proof of the main theorem is the following proposition, which has been established for the massive system in \cite{a072}. It is generalized here to the massless Einstein-Vlasov system and slightly improved in the sense that a larger class of ansatz functions is considered.

\begin{prop}\label{main_theorem}
Consider a quasi shell solution of the massless Einstein-Vlasov system with a sufficiently small inner radius of support $R_0$. The distribution function $f$ then vanishes within the interval
\begin{equation}
\left[R_0\left(1+B_0R_0^{\frac{2}{q+1}}\right),R_1\right], \quad \mathrm{\it where}\;R_1 = R_0 \frac{1+B_2 R_0^{\frac{1}{q+1}}}{1-B_1 R_0^{\frac{2}{q+1}}},
\end{equation}
and $B_0$, $B_1$, $B_2$, and $q$ are positive constants which only depend on the ansatz function for the matter distribution. The solution can be joined with a Schwarzschild solution at the point where $f$ vanishes and a static shell is obtained with support within $[R_0,R_1]$. The mass parameter of this Schwarzschild solution corresponds to the mass of the matter shell and the metric components $\mu$ and $\lambda$ as well as the matter quantities (\ref{mat11}) -- (\ref{mat13}) are continuously differentiable.
\end{prop}

\begin{rem}
Note that these shell solutions can be made infinitely thin as they are moved towards the origin, i.e.~$\frac{R_1}{R_0}\to 1$ as $R_0\to 0$.
\end{rem}


\section{Proof of the main theorem} \label{proofs}

We verify that the corresponding result to \cite{a072} in the massive case holds similarly in the massless case. The result is slightly refined to cover a larger class of ansatz functions. For the sake of completeness we give the entire adapted proof in detail. Eventually we show that a gluing method yields solutions of the claimed regularity.

\subsection{Bounds on the matter quantities}
A precise control of the matter quantities will be crucial for the proof of the main theorem. First we derive explicit bounds for the matter quantities in terms of $\gamma$ and $r$.
\begin{lemma} \label{lem_matter}
Let $R_0$ be defined as in (\ref{defr0}) and $\gamma$ as in (\ref{defgamma}). Moreover, let $k\geq 0$ and $\ell \geq -\frac 12$. Then for $r\in[R_0,2R_0]$ there are positive constants $C_{l}$ and $C_u$ depending on $k$, $\ell$, and $L_0$ such that when $\gamma \geq 0$,
\begin{align}
C_l \frac{\gamma^{\ell + k + \frac 3 2}}{r^4} \leq&\, \varrho(r) \leq C_u \frac{\gamma^{\ell + k + \frac 3 2}}{r^4}, \label{bounds_rho} \\
C_l \frac{\gamma^{\ell + k + \frac 5 2}}{r^4} \leq&\, p(r) \leq C_u \frac{\gamma^{\ell + k + \frac 5 2}}{r^4}, \label{bounds_p} \\
C_l \frac{\gamma^{\ell + k + \frac 3 2}}{r^2} \leq&\, z(r) \leq C_u \frac{\gamma^{\ell + k + \frac 3 2}}{r^2}, \label{bounds_z}
\end{align}
where $z = \varrho - p -2p_T$ and $\varrho$, $p$, and $p_T$ are as in (\ref{mat11}) -- (\ref{mat13}).
\end{lemma}

\begin{proof}
We will only carry out the proof in detail for (\ref{bounds_rho}) since the other two bounds, (\ref{bounds_p}) and (\ref{bounds_z}), follow in a very similar way. By inserting the ansatz (\ref{ansf}) into the expression (\ref{mat11}) for $\varrho$ we obtain the form (\ref{mat21}). We consider the integral over $\varepsilon$ in more detail. First, we rewrite it in a convenient way. Observe that by definition of $\gamma$ given in (\ref{defgamma}) we have
\begin{equation}
e^{y} = e^{\gamma}\sqrt{\frac{L_0}{r^2}}.
\end{equation}
We substitute this and set in addition $a = \sqrt{L_0/r^2}$ and $b = e^\gamma$. We obtain
\begin{equation}
\varrho(r) = 2\pi c_\ell r^{2\ell} \int_a^{ba} \left[1-\frac{\varepsilon}{ba}\right]^k \varepsilon^2 \left(\varepsilon^2-a^2\right)^{\ell + \frac 1 2} \mathrm d\varepsilon.
\end{equation}
A further change of variables given by $t=\frac{\varepsilon/a-1}{b-1}$ in the integral yields
\begin{equation}
\varrho(r) = 2\pi c_\ell \frac{L_0^{\ell + 2}}{r^4} b^{-k}(b-1)^{k+\ell+\frac{3}{2}} \int_0^1 t^{\ell + \frac 1 2} (1-t)^k (t(b-1)+1)^2 (t(b-1)+2)^{\ell + \frac 1 2} \mathrm dt.
\end{equation}
We call this integral $I_{k,\ell}\left(e^\gamma\right)$. Note that $e^\gamma = b$ and $I_{k,\ell}\left(e^\gamma\right)$ is always positive since $b\geq 1$ in view of the assumption $\gamma \geq 0$. Moreover, the integral has bounds from below and from above given by
\begin{equation}
I_{k,\ell} \geq \int_{\frac{1}{4}}^{\frac{1}{2}} \left(\frac{1}{4}\right)^{\ell + \frac 1 2} \left(\frac 1 2\right)^k 2^{\ell + \frac 1 2} \mathrm dt = 2^{-(k+\ell + 3)}
\end{equation}
and
\begin{equation} \label{int_upper_b}
I_{k,\ell} \leq \int_0^1 b^2 (b+1)^{\ell + \frac 1 2}\mathrm dt = b^2(b+1)^{\ell + \frac 1 2}.
\end{equation}
The last quantity (\ref{int_upper_b}) will of course be bounded if $b$ is bounded. We will observe in the following that this is in fact the case for $r\in[R_0,2R_0]$. Note that these bounds do not depend on $R_0$ if this radius is considered as being determined by a choice of $y_0$. We have obtained that $\varrho$ is of the form
\begin{equation} \label{rho_pol}
\varrho(r) = C_{k,\ell}^\varrho \frac{L_0^{2\ell}}{r^4} e^{-k\gamma} \left(e^\gamma -1\right)^{\ell + k + \frac 3 2} I_{k,\ell}\left(e^\gamma\right)
\end{equation}
with a positive constant $C_{k,\ell}^\varrho$ only depending on $k$ and $\ell$. \par
In the remainder of the proof we make use of the assumption that $r\in[R_0,2R_0]$. We calculate the derivative of the function $\gamma(r)$ and obtain
\begin{equation} \label{gammas}
\gamma'(r) = y'(r)+\frac{1}{r} = -e^{2\lambda} \left[\frac{m(r)}{r^2} + 4\pi rp(r)\right] + \frac{1}{r}.
\end{equation}
Note that in particular $\gamma'(r) \leq \frac 1 r$. Since $\gamma(R_0)=0$ this implies that
\begin{equation}
\gamma(r) \leq \ln\left(r/R_0\right) \leq \ln(2) < 1
\end{equation}
 for $r\in [R_0,2R_0]$. As an immediate consequence it follows that the factor $e^{-k\gamma}$ is bounded from below and from above, namely we have $e^{-k} \leq e^{-k\gamma} \leq 1$. Finally, we want to estimate the expressions of the form $\left(e^\gamma -1\right)^{\ell + k + \frac 3 2}$ by powers of $\gamma$. In view of the Taylor expansion
\begin{equation}
e^\gamma = \sum_{j=0}^\infty \frac{\gamma^j}{j!} \leq 1 + \gamma + \frac 1 2 \sum_{j=2}^{\infty} \gamma^j = 1 + \gamma + \frac{\gamma^2}{2(1-\gamma)}
\end{equation}
we have 
\begin{equation}
\gamma(r) \leq \left[e^{\gamma(r)}-1\right] \leq C \gamma(r)
\end{equation}
for a positive constant $C$ since $\gamma(r) \leq \ln(2)$ on $[R_0,2R_0]$. For the other matter quantities similar bounds can be derived using the form (\ref{mat22}) for $p$ and a similar form for $p_T$. The corresponding expression for $p_T$ obtained by the same succession of variable transformations is less compact than (\ref{mat21}) and (\ref{mat22}) but the structure is similar. In \cite{rein93} this form is derived explicitly. Hence one can take $C_l$ as the minimum of these constants and $C_u$ as their maximum. 
\end{proof}

\subsection{Estimates on the characteristic function}

In the remainder of this work it will be useful to have the abbreviation
\begin{equation} \label{def_q}
q = \ell + k + \frac 3 2 \geq 1
\end{equation}
at hand. A direct consequence of Lemma \ref{lem_matter} is that on the interval $[R_0,2R_0]$ the characteristic function $\gamma(r)$ can be estimated by a power of $r$ as stated by the following lemma.

\begin{lemma} \label{lem_est_gamma}
Given $k \geq 0$ and $\ell \geq -\frac 1 2$ there is a constant $C_\gamma > 0$ such that
\begin{equation}
\gamma(r) \leq C_\gamma r^{\frac{2}{q+1}}
\end{equation}
for all $r\in[R_0,2R_0]$, if $R_0$ is chosen sufficiently small.
\end{lemma}

\begin{proof}
We show that for $r\in[R_0,2R_0]$
\begin{equation} \label{est_gamma}
\gamma(r)^{q+1} \leq \frac{r^2}{4\pi C_l}
\end{equation}
by a contradiction argument. $C_l$ is the constant in Lemma \ref{lem_matter}. Recall that by its definition (\ref{defgamma}) we have $\gamma(R_0) = 0$, so for $r=R_0$ the inequality (\ref{est_gamma}) is trivially fulfilled. Assume that there is a subinterval $[r_1,r_2]\subset [R_0,2R_0]$ such that
\begin{equation*}
\gamma(r_1)^{q+1} = \frac{r_1^2}{4\pi C_l}, \quad \mathrm{and} \quad \gamma(r)^{q+1} > \frac{r^2}{4\pi C_l} \quad \mathrm{for}\;r\in(r_1,r_2).
\end{equation*}
Note that by virtue of Lemma \ref{lem_matter} we have 
\begin{equation}
p(r_1) \geq C_l\frac{\gamma(r_1)^{q+1}}{r_1^4} = \frac{1}{4\pi r_1^2}.
\end{equation}
Consider $\gamma'$ given in (\ref{gammas}) at $r=r_1$:

\begin{equation}
\gamma'(r_1) \leq -\left[\frac{m(r_1)}{r_1^2} + \frac{1}{r_1}\right] e^{2\lambda} + \frac{1}{r_1} \leq -\frac{m(r_1)}{r_1^2}e^{2\lambda} \leq 0.
\end{equation}
But on the other hand
\begin{equation}
\gamma(r)^{q+1} > \frac{r^2}{4\pi C_l}
\end{equation}
 requires $\gamma'(r_1) > 0$ which yields a contradiction. So on $[R_0,2R_0]$ we have 
\begin{equation}
\gamma(r)^{q+1} \leq \frac{r^2}{4\pi C_l} \quad \Leftrightarrow \quad \gamma(r)^{q+1} \leq (4\pi C_l)^{-\frac{1}{q+1}} r^{\frac{2}{q+1}} =: C_\gamma r^{\frac{2}{q+1}}
\end{equation}
and the proof is complete.
\end{proof}

We proceed with a lower bound on the derivative of the characteristic function.

\begin{lemma} \label{lem_gamma_p}
Let $C_m=\mathrm{max}\{1,C_u,4\pi C_u\}$ and let 
\begin{equation}
\delta \leq \frac{R_0^{\frac{q+3}{q+1}}}{C_m8^{\frac{1}{q+1}}}.
\end{equation}
Then
\begin{equation}
\gamma'(r) \geq \frac{1}{2r}
\end{equation}
for $r\in [R_0,R_0+\delta]$ if $R_0$ is chosen small enough.
\end{lemma}

\begin{proof}
First we observe that for $R_0$ small enough we have $\delta < 2R_0$ which enables us to use Lemma \ref{lem_matter}. We have $\gamma'(r) = -y'(r) + \frac{1}{r}$. Let $\sigma \in [0,\delta]$. Since $\gamma(R_0)=0$, the mean value theorem implies that 
\begin{equation}
\gamma(R_0+\sigma) = \gamma(R_0+\sigma) - \gamma(R_0) \leq \sigma\gamma'(\xi) \leq \frac{\delta}{R_0}
\end{equation}
 for a $\xi \in [R_0,R_0+\sigma]$. This implies $\gamma(r) \leq \frac{\delta}{R_0}$. We substitute this in the upper bound (\ref{bounds_rho}) for $\varrho$ stated in Lemma \ref{lem_matter} and obtain
\begin{equation}
\varrho \leq C_u \frac{\delta^q}{R_0^q r^4}.
\end{equation}
Since $\varrho = 0$ when $r < R_0$ we obtain for $\sigma\in[0,\delta]$
\begin{align}
&m(R_0+\sigma) \leq \int_{R_0}^{R_0+\sigma} s^2 \frac{C_u \delta^q}{R_0^q s^4} \mathrm ds = \frac{C_u^\varrho \delta^q }{R_0^q}\frac{\sigma}{R_0(R_0+\sigma)}\nonumber \\
\Rightarrow \quad  &\frac{m(R_0+\sigma)}{R_0+\sigma} \leq \frac{C_u \delta^{q}}{R_0^{q+1}} \frac{\sigma}{(R_0+\sigma)^2} \leq \frac{C_u \delta^{q+1}}{R_0^{q+3}}.
\end{align}
We also have for the pressure $p$ the upper bound given in (\ref{bounds_p}) and therefore the estimate
\begin{equation}
p = C_u \frac{\gamma^{q+1}}{r^4} \leq \frac{C_u\delta^{q+1}}{r^4 R_0^{q+1}}, \quad \mathrm{for}\; r\in[R_0,R_0+\delta].
\end{equation}
Since by assumption $\delta \leq \frac{R_0^{\frac{q+3}{q+1}}}{C_m8^{\frac{q}{q+1}}}$ and  $C_m \geq C_u$ we have 
\begin{equation}
\delta^{q+1} \leq \frac{R_0^{q+3}}{8C_m^{q+1}}\leq \frac{R_0^{q+3}}{8C_m}
\end{equation}
and thus 
\begin{equation}
\frac{m(R_0+\sigma)}{R_0+\sigma} \leq \frac{C_u}{8C_m} \leq  \frac 1 8.
\end{equation}
Hence
\begin{equation}
e^{2\lambda(R_0+\sigma)} = \frac{1}{1-\frac{2m(R_0+\sigma)}{R_0+\sigma}} \leq \frac{1}{1-\frac{1}{4}} = \frac{4}{3}
\end{equation}
and also
\begin{equation}
4\pi r^2 p \leq \frac{4\pi}{r^2} \frac{C_u \delta^{q+1}}{R_0^{q+1}} \leq \frac{4\pi}{r^2} \frac{C_u R_0^{q+3}}{8 R_0^{q+1} C_m^{q+1}} \leq \frac{R_0^2}{8r^2} \leq \frac{1}{8}.
\end{equation}
Thus for $r\in[R_0,R_0+\delta]$ we have obtained
\begin{equation}
ry'(r) = -e^{2\lambda(r)} \left[4\pi r^2 p(r) + \frac{m(r)}{r}\right] \geq -\frac{1}{3}
\end{equation}
and so we have
\begin{equation}
\gamma'(r) = y'(r) + \frac{1}{r} = \frac{1}{r}\left(y'(r)+1\right) \geq \frac{1}{r} \frac 2 3 \geq \frac{1}{2r}
\end{equation}
as desired.
\end{proof}

\subsection{Existence of characteristic radii}
The previous lemma reveals that the quantity $\gamma$ increases relatively steeply after the inner vacuum region. In the following we show that this quantity will also decrease in a controlled way and eventually reach zero. To formulate this, some constants are defined. By virtue of Lemma \ref{lem_gamma_p} we have for $\sigma\in[0,\delta]$ the estimate
\begin{align*}
\gamma(R_0+\sigma) \geq \gamma(R_0) + \sigma \inf_{s\in[0,\sigma]} \gamma'(R_0+s) \geq 
\frac{\sigma}{2(R_0+\sigma)}.
\end{align*}
Now let
\begin{equation}
\sigma_*:=C_0R_0^{1+\frac{2q+1}{q(q+1)}},
\end{equation}
where $C_0 = \min\left\{ \frac 1 4, \left(C_m 8^{\frac{1}{q+1}}\right)^{-1}\right\}$. Since
\begin{equation}
1 + \frac{2q+1}{q(q+1)} \geq \frac{q+3}{q+1} 
\end{equation}
we have by construction
\begin{equation} \label{def_sigmas}
\sigma_* = C_0R_0^{1+\frac{2q+1}{q(q+1)}} \leq \frac{1}{C_m 8^{\frac{1}{1+q}}} R_0^{1 + \frac{2q+1}{q(q+1)}} \leq \frac{R_0^{\frac{q+3}{q+1}}}{C_m 8^{\frac{1}{1+q}}}
\end{equation}
since $R_0$ can be chosen small. Since $\sigma_* \sim R_0^{1+\frac{2q+1}{q(q+1)}}$ and $\delta \lesssim R_0^{\frac{q+3}{q+1}}$ one has $\sigma_* < \delta$ for $R_0$ small enough. We define
\begin{equation}
\gamma_* := \frac{\sigma_*}{2(R_0+\sigma_*)}.
\end{equation}
One easily checks that 
\begin{equation}
\gamma(R_0+\sigma_*) \geq (R_0 + \sigma_*) \frac{1}{2(R_0+\delta)} \geq \gamma_*.
\end{equation}
 Since $\gamma$ started at $\gamma(R_0) = 0$, it has reached the value $\gamma_*$ already at least once. We call this radius $r_1$. The following lemma states that this value $\gamma_*$ will be attained a second time shortly after the radius $R_0+\sigma_*$.

\begin{lemma} \label{lem_gammas}
Let $\kappa = \frac{9 2^{2q}}{C_l C_0^q}$ and let $K = \max\{C_0,\kappa\}$, and consider a solution with $R_0$ such that
\begin{equation} \label{r0_small}
R_0^{\frac{1}{1+q}}K \leq 1.
\end{equation}
Then there is a point $r_2$ such that $\gamma(r) \searrow \gamma_*$ as $r \nearrow r_2$, and $$r_2 \leq R_0 + \sigma_* + \kappa R_0^{\frac{q+2}{q+1}}.$$
\end{lemma}

\begin{proof}
Since $\gamma(R_0+\sigma_*)\geq \gamma_*$, and since Lemma \ref{lem_gamma_p} gives that $\gamma'(R_0+\sigma_*) \geq \frac{1}{2(R_0+\sigma_*)} > 0$, the radius $r_2$ must be strictly larger than $R_0+\sigma_*$. Let $\Delta > 0$ such that $\gamma \geq \gamma_*$ on the interval $[R_0+\sigma_*, R_0+\sigma_*+\Delta]$ and also $[R_0+\sigma_*, R_0+\sigma_*+\Delta] \subset [R_0,2R_0]$. We will show that
\begin{equation}
\Delta \leq \kappa R_0^{\frac{q+2}{q+1}}
\end{equation}
which implies the assertion. Since Lemma \ref{lem_matter} provides on $[R_0,2R_0]$ the estimate $\varrho \geq C_l \frac{\gamma^q}{r^4}$ we have for $r\in [R_0+\sigma_*+\Delta,2R_0]$
\begin{equation}
\begin{aligned}
m(r) &= 4\pi C_l \int_{R_0}^r \frac{\gamma^q(s)}{s^2} \mathrm ds\\
& \geq 4\pi C_l \int_{R_0}^{R_0+\sigma_*} \frac{\gamma^q(s)}{s^2}\mathrm ds + 4\pi C_l \int_{R_0+\sigma_*}^{R_0+\sigma_*+\Delta} \frac{\gamma^q(s)}{s^2} \mathrm ds  \\
&\geq 4\pi C_l \int_{R_0+\sigma_*}^{R_0+\sigma_*+\Delta} \frac{\left(\gamma^*\right)^q}{s^2}\mathrm ds\\
& \geq C_l \int_{R_0+\sigma_*}^{R_0+\sigma_*+\Delta} \frac{\sigma_*^q}{2^q(R_0+\sigma_*)^q s^2}\mathrm ds  \\
&= C_l \frac{\sigma_*^q}{2^q(R_0+\sigma_*)^{q+1}}\frac{\Delta}{R_0+\sigma_*+\Delta}.
\end{aligned}
\end{equation}
Hence
\begin{equation} \label{ineq_mrsd}
\frac{m(R_0+\sigma_*+\Delta)}{R_0+\sigma_*+\Delta} \geq C_l \frac{\sigma_*^q}{2^q(R_0+\sigma_*)^{q+1}} \frac{\Delta}{(R_0+\sigma_*+\Delta)^2}.
\end{equation}
The assumption (\ref{r0_small}) guarantees $\sigma_* \leq R_0$. Inserting this into the upper inequality (\ref{ineq_mrsd}) yields
\begin{equation}\begin{aligned}
\frac{m(R_0+\sigma_*+\Delta)}{R_0+\sigma_*+\Delta} &\geq C_l \frac{C_0^q R_0^{q\left(1+\frac{2q+1}{q(q+1)}\right)}}{2^q 2^{q+1} R_0^{q+1}} \frac{\Delta}{(R_0+\sigma_*+\Delta)^2}\\
 &=\frac{C_l C_0^q  R_0^{\frac{q}{q+1}}\Delta}{2^{2q+1}(R_0+\sigma_*+\Delta)^2}.
\end{aligned}
\end{equation}
Assume now that the assertion of the lemma does not hold, i.e.~in particular $\Delta \geq \kappa R_0^{\frac{q+2}{q+1}}$. On the radial axis we of course have to remain in the domain $R_0 + \sigma_* + \Delta < 2R_0$ but (\ref{r0_small}) guarantees that this is always possible. It follows (also by $\sigma_* \leq R_0$)
\begin{equation}
\begin{aligned}
\frac{m(R_0 + \sigma_* + \Delta)}{R_0 + \sigma_* + \Delta}  &= C_l C_0^q \frac{R_0^{\frac{q}{q+1}}\kappa R_0 R_0^{\frac{1}{1+q}}}{2^{2q+1}(R_0 + \sigma_* + \Delta)^2}\\
& \geq C_l C_0^q \kappa \frac{R_0^2}{2^{2q+1}(3R_0)^2}\\
& = \frac{C_l C_0^q \kappa}{2^{2q+1}3^2}.
\end{aligned}
\end{equation}
And the definition of $\kappa$ implies that
\begin{equation}
\frac{m(R_0+\sigma_*+\Delta)}{R_0 + \sigma_* + \Delta} \geq \frac{C_l C_0^{q}\kappa}{3^2 2^{2q+1}} = \frac 1 2.
\end{equation}
But this leads to a contradiction since it is known that all static solutions fulfill $\frac{m(r)}{r} \leq \frac{4}{9} < \frac 1 2$ \cite{a08}. So we deduce $$\Delta < \kappa R_0^{\frac{q+2}{q+1}}$$ and the proof is complete.
\end{proof}

\subsection{A lower bound for $m/r$}
By now we have characterized the behavior of the quantity $\gamma(r)$ in some detail. Piling on those lemmas we can prove the following proposition giving a lower bound on the quantity $\frac{m}{r}$ at the radius $r_2$ where $\gamma(r_2)=\gamma_*$ a second time.

\begin{prop} \label{prop_m_r}
Let $r_2$ be as in Lemma \ref{lem_gammas} for a sufficiently small $R_0$. Then the corresponding solution satisfies
\begin{equation}
\frac{m(r_2)}{r_2} \geq \frac{2}{5}.
\end{equation}
\end{prop}

\begin{proof}
In the first part of the prove we investigate the expression $m(r_2)e^{(\lambda-y)(r_2)}$. First note that the Einstein equations (\ref{eeq1}) and (\ref{eeq2}) imply
\begin{equation}
\lambda' - y' = 4\pi r e^{2\lambda}(\varrho + p).
\end{equation}
Using this and the TOV equation (\ref{tov_equation}) we can calculate the $r$-derivative of the expression $e^{\lambda-y}\left(m(r)+4\pi r^3 p(r)\right)$ and obtain
\begin{equation}
\frac{\mathrm d}{\mathrm dr} e^{\lambda-y}\left(m(r)+4\pi r^3 p(r)\right) = 4\pi r^2 e^{\lambda-y} (\varrho + p +2p_T).
\end{equation}
So we have
\begin{equation}
e^{\lambda-y}\left[\frac{m}{r^2} + 4\pi r p\right] = \frac{1}{r^2} \int_0^r 4\pi s^2 e^{(\lambda-y)(s)}(\varrho + p + 2p_T)\mathrm ds.
\end{equation}
We consider this equality at $r=r_2$. Furthermore we substitute $z = \varrho - p -2p_T$ and use the fact that all matter quantities are zero for $r \leq R_0$. We obtain
\begin{equation} \label{mr2_exp2}
m(r_2) e^{(\lambda-y)(r_2)} = \int_{R_0}^{r_2} 4\pi s^2 e^{(\lambda-y)(s)}(2\varrho - z)\mathrm ds - e^{(\lambda-y)(r_2)}4\pi r_2^3 p(r_2).
\end{equation}
Lemma \ref{lem_matter} implies 
\begin{equation}z \leq C_u \frac{\gamma^{\ell+k+\frac 3 2}}{r^2} \leq \frac{C_u}{C_l} r^2 C_l \frac{\gamma^{\ell+k+\frac 3 2}}{r^4} \leq\frac{C_u}{C_l} r^2 \varrho(r).
\end{equation} 
Inserting this relation into (\ref{mr2_exp2}) yields
\begin{equation} \label{mr2_exp3}
\int_{R_0}^{r_2} 4\pi s^2 e^{(\lambda-y)(s)} z(s)\;\mathrm ds \leq 4\pi C r_2^3 e^{-y(r_2)} \int_{R_0}^{r_2} s e^{\lambda(s)} \varrho(s) \mathrm ds
\end{equation}
for a positive constant $C$. The $y$-factor can be taken in front of the integral since $y$ is monotonically decreasing. In the remainder of the proof $C$ shall be understood as a symbolic notation for a positive constant whose value may change from line to line. \par
We estimate the integral in (\ref{mr2_exp3}) further. Observe that the first Einstein equation (\ref{eeq1}) implies $4\pi r\varrho e^\lambda = -\frac{\mathrm d}{\mathrm dr}\left(e^{-\lambda}\right) + e^\lambda \frac{m}{r^2}$. We write
\begin{align}
\int_{R_0}^{r_2} 4\pi s \varrho(s) e^{\lambda(s)} \mathrm ds &= \int_{R_0}^{r_2}\left[-\frac{\mathrm d}{\mathrm ds} e^{-\lambda}\right]\mathrm ds + \int_{R_0}^{r_2} \frac{m(s) e^{\lambda(s)}}{s^2}\mathrm ds \nonumber \\
&\leq 1-\sqrt{1-\frac{2m(r_2)}{r_2}} + C m(r_2) R_0^{\frac{q+2}{q+1}-2} \nonumber \\
&\leq \frac{2m(r_2)}{R_0} + \frac{C m(r_2)}{R_0^{\frac{q}{q+1}}} \leq C\frac{2m(r_2)}{R_0}. \label{upper_b_int_r_l}
\end{align}
Here it was used that $e^{-2\lambda} = 1-\frac{2m}{r}$, $m(r)$ is increasing in $r$, and that
\begin{equation}
r_2 \leq R_0 + \sigma_* + \kappa R_0^{\frac{q+2}{q+1}} \Rightarrow r_2-R_0 \leq C R_0^{\frac{q+2}{q+1}}
\end{equation}
for a positive constant $C$ (cf.~also the definitions (\ref{def_sigmas}) and (\ref{def_q}) and $\kappa$ is defined in Lemma \ref{lem_gammas}). \par
The next step of the proof is to show that $r_2 \leq 4m(r_2)$. Therefore we first observe that $\gamma'(r_2) \leq 0$ since at $r_2$, $\gamma(r)$ approaches $\gamma_*$ from above as stated in Lemma \ref{lem_gammas} and also
\begin{equation} \label{est_p_s}
r_2 p(r_2) \leq C_u \frac{\gamma_*^{q+1}}{r_2^3} \leq \frac{C}{R_0^{\frac{q-1}{q}}}
\end{equation}
by Lemma \ref{lem_matter}. By the definition (\ref{defgamma}) of $\gamma$ we have $y'(r_2) = \gamma'(r_2) - \frac{1}{r_2} \leq -\frac{1}{r_2}$. Now assume that $4m(r_2) < r_2$. Then $e^{2\lambda} \leq 2$ and we get
\begin{equation}
y'(r_2) = -e^{2\lambda}\left[\frac{m(r_2)}{r_2^2} + 4\pi r_2 p(r_2) \right] \geq -2\left[\frac{1}{4 r_2} + \frac{4\pi C}{R_0^{\frac{q-1}{q}}}\right].
\end{equation}
This expression is certainly larger than $-r_2^{-1}$ if $R_0$ is chosen small enough, since a small $R_0$ also implies a small $r_2$ by its definition in Lemma \ref{lem_gammas}. This is a contradiction to the previous simple observation. Hence $r_2 \leq 4m(r_2)$. \par
Coming back to the expression (\ref{mr2_exp2}) we apply the estimates (\ref{mr2_exp3}) and (\ref{upper_b_int_r_l}) and use $r_2 \leq 4m(r_2)$. We have
\begin{equation}
\begin{aligned}
&m(r_2) e^{(\lambda-y)(r_2)}\\
 &\geq \int_{R_0}^{r_2} 4\pi s^2 e^{(\lambda-y)(s)} 2\varrho(s)\mathrm ds - C r_2^3 e^{-y(r_2)}\int_{R_0}^{r_2} 4\pi s \varrho(s) e^\lambda \mathrm ds - 4\pi r_2^3 p e^{(\lambda-y)(r_2)} \\
& \geq 2\int_{R_0}^{r_2} 4\pi s^2 e^{(\lambda-y)(s)} \varrho(s) \mathrm ds - C e^{-y(r_2)} r_2^2 m(r_2) - C r_2^2 m(r_2) p e^{(\lambda-y)(r_2)}  \\
& \geq 2 e^{-y(R_0)} R_0 \int_{R_0}^{r_2} 4\pi s e^{\lambda(s)} \varrho(s)\mathrm ds - C R_0^2 m(r_2)\left(e^{-y(r_2)} + p(r_2) e^{(\lambda - y)(r_2)}\right) \label{mr2_exp4}.
\end{aligned}
\end{equation}
We have already observed that the first Einstein equation (\ref{eeq1}) implies $4\pi \varrho e^\lambda = -\frac{\mathrm d}{\mathrm dr}(e^{-\lambda}) + e^\lambda \frac{m}{r^2}$. We estimate
\begin{align}
\int_{R_0}^{r_2} 4\pi s \varrho(s) e^{\lambda(s)} \mathrm ds &= \int_{R_0}^{r_2} \left[-\frac{\mathrm d}{\mathrm ds} e^{-\lambda}\right] \mathrm ds + \int_{R_0}^{r_2} \frac{m(s) e^{\lambda(s)}}{s^2}\mathrm ds \nonumber \\
&\geq 1- \sqrt{1-\frac{2m(r_2)}{r_2}} = \frac{2m(r_2)}{r_2\left[1+\sqrt{1-\frac{2m(r_2)}{r_2}}\right]}.
\end{align}
Substituting this in (\ref{mr2_exp4}) we obtain
\begin{equation}
m(r_2) e^{(\lambda - y)(r_2)} \geq e^{-y(R_0)} \frac{R_0}{r_2} \frac{4m(r_2)}{1+\sqrt{1-\frac{2m(r_2)}{r_2}}} -C  R_0^2 m(r_2) \left(e^{-y(r_2)} + p(r_2) e^{(\lambda - y)(r_2)}\right).
\end{equation}
Again we use the estimate (\ref{est_p_s}) for $p(r_2)$ and the facts that $e^\lambda \geq 1$ and that $y(r)$ decreases in $r$. We get
\begin{align}
m(r_2) e^{(\lambda - y)(r_2)} &\geq e^{-y(R_0)} \frac{R_0}{r_2} \frac{4m(r_2)}{1+\sqrt{1-\frac{2m(r_2)}{r_2}}}-C e^{-y(r_2)} R_0^2 m(r_2) - CR_0^{\frac 1 q} m(r_2) e^{(\lambda - y)(r_2)} \nonumber \\
\Rightarrow \quad e^{y(R_0)-y(r_2)} &\geq \frac{R_0}{r_2} \frac{4 e^{-\lambda(r_2)}}{1+\sqrt{1-\frac{2m(r_2)}{r_2}}}-C e^{y(R_0)-y(r_2)} R_0^2 - C R_0^{\frac 1 q} e^{y(R_0)-y(r_2)} \nonumber \\
\Rightarrow \quad 1 &\geq e^{y(r_2)-y(R_0)} \frac{R_0}{r_2} \frac{4 e^{-\lambda(r_2)}}{1+\sqrt{1-\frac{2m(r_2)}{r_2}}} - CR_0^2 - CR_0^{\frac 1 q}. \label{mr2_exp5}
\end{align}
Furthermore we have
\begin{equation} \label{int_to_0}
y(r_2)-y(R_0) = - \int_{R_0}^{r_2} e^{2\lambda(s)} \left[4\pi s p(s) + \frac{m(s)}{s^2}\right] \mathrm ds.
\end{equation} \par
We show that this integral (\ref{int_to_0}) goes to zero as $R_0$ goes to zero. From the estimate in Lemma \ref{lem_est_gamma} it follows that $p(r) \leq \frac{C}{r^2}$ for a positive constant $C$. Moreover it is known that for every static solution of the Einstein Vlasov system we have the inequality $\frac{2m(r)}{r} \leq \frac{8}{9}$, \cite{a08}. This implies that $e^{2\lambda} = \left[1-\frac{2m(r)}{r}\right]^{-1}$ is bounded from above by 9. We get for the integrand of \eqref{int_to_0}
\begin{equation}
\left[\frac{m}{r^2} + 4\pi r p\right] e^{2\lambda} \leq 9\left[\frac{4}{9r} + \frac{4\pi C}{r}\right] =: \frac{C_p}{r}.
\end{equation}
Hence
\begin{equation}
y(r_2)-y(R_0) \geq -\int_{R_0}^{r_2} \frac{C_p}{s} \mathrm ds = [-C_p \ln(s)]_{R_0}^{r_2} = -C_p\ln\left[\frac{r_2}{R_0}\right].\end{equation}
This implies that (\ref{mr2_exp5}) can be written as
\begin{equation}
1 \geq \left[\frac{R_0}{r_2}\right]^{C_p + 1} \frac{4\sqrt{1-\frac{2m(r_2)}{r_2}}}{1+\sqrt{1-\frac{2m(r_2)}{r_2}}} - CR_0^2 - CR_0^{\frac{1}{q}}.
\end{equation}
Since $\frac{R_0}{r_2} \nearrow 1$, as $R_0 \to 0$ we can write this inequality as
\begin{equation}
1 \geq (1-\Gamma(R_0))\frac{4\sqrt{1-\frac{2m(r_2)}{r_2}}}{1+\sqrt{1-\frac{2m(r_2)}{r_2}}} - C\Gamma(R_0)
\end{equation}
with a function $\Gamma(R_0)$ fulfilling $0 < \Gamma(R_0) \to 0$, $R_0\to 0$. This yields
\begin{equation}
1+\sqrt{1-\frac{2m(r_2)}{r_2}} \geq (1-\Gamma(R_0)) 4 \sqrt{1-\frac{2m(r_2)}{r_2}} - C\Gamma(R_0),
\end{equation}
so that
\begin{equation}
1 \geq 3\sqrt{1-\frac{2m(r_2)}{r_2}} - C\Gamma(R_0).
\end{equation}
We square this inequality and obtain
\begin{align}
1 + 2C\Gamma(R_0) + C^2\Gamma^2(R_0) &\geq 9\left[1-\frac{2m(r_2)}{r_2}\right] \nonumber \\
\Rightarrow \quad \frac{2m(r_2)}{r_2} \geq \frac{8}{9} - C\Gamma(R_0).
\end{align}
By choosing $R_0$ small one can let the quotient $\frac{m(r_2)}{r_2}$ become arbitrarily close to $\frac 4 9$, in particular one can attain $\frac{m(r_2)}{r_2} \geq \frac{2}{5}$ which completes the proof of the proposition.
\end{proof}

\subsection{Proof of Proposition \ref{main_theorem}}

\begin{proof}[Proof of Proposition \ref{main_theorem}]
The matter quantities are different from zero if and only if $\gamma(r) > 0$. Since $\gamma(R_0) = 0$ and on $[R_0,R_0+\delta]$ we have $\gamma'(r) \geq \frac{1}{2r} > 0$. By virtue of Lemma \ref{lem_gammas} we can observe that $f$ cannot vanish for
\begin{equation}
r \leq R_0 + \frac{R_0^{\frac{q+3}{q+1}}}{C_m 8^{\frac{1}{1+q}}}.
\end{equation}
Thus the claim that the matter quantities do not vanish before $r=R_0 + B_0R_0^{\frac{q+3}{q+1}}$ follows with $B_0 = \frac{1}{C_m 8^{\frac{1}{q+1}}}$. \par
The main issue in this proof is to show that the matter quantities vanish before some radius
\begin{equation}
R_1 = \frac{R_0+B_2R_0^{\frac{q+2}{q+1}}}{1-B_1R_0^{\frac{2}{q+1}}}
\end{equation}
where $B_1$ and $B_2$ are positive constants. We will show that $\gamma$ necessarily vanishes close to the point $r_2$ if $R_0$ is sufficiently small. For this purpose the variable
\begin{equation} \label{def_x}
x = \frac{m(r)}{r \gamma(r)}
\end{equation}
is defined. Observe that $x \to \infty$ implies $\gamma \to 0$ and therefore also $f \to 0$. So we want to prove that $x$ diverges at a finite radius. We consider
\begin{equation} \label{an_est_x}
rx' = \frac{4\pi r^2\varrho}{\gamma} - x + \frac{x^2}{1-2\gamma x} - \frac{x}{\gamma} + \frac{4\pi r^3 p x^2}{1-2\gamma x}.
\end{equation}
for $r \in [R_0,r_2]$ and $\gamma(r) > 0$ and we will show that $\gamma(r)  = 0$ for some $r < (1+\Gamma(R_0))r_2$, where $\Gamma$ is some function having the property that $0 < \Gamma(R_0) \rightarrow 0$, as $R_0 \to 0$. Since $\gamma > 0$ and $\varrho, p \geq 0$ the first term and the last term in (\ref{an_est_x}) can be dropped and we have
\begin{equation} \label{rxp2}
rx' \geq \frac{x^2}{3(1-2\gamma x)} -x + \frac{2 x^2}{3(1-2\gamma x)} - \frac{x}{\gamma}.
\end{equation}
Take $R_0$ sufficiently small so that $\frac{m(r_2)}{r_2} \geq \frac{2}{5}$ by Proposition \ref{prop_m_r}. Let $r\in\left[r_2, \frac{16 r_2}{15}\right]$, then since $m$ is increasing in $r$ we get
\begin{equation} \label{est_m_r}
\frac{m(r)}{r} \geq \frac{m(r_2)}{r} = \frac{r_2}{r} \frac{m(r_2)}{r_2} \geq \frac{3}{8}.
\end{equation}
We which to estimate $rx' $ in (\ref{rxp2}) further. Observe that by (\ref{est_m_r}) and (\ref{def_x}) we have
\begin{equation}
\frac{x}{1-2\gamma x} = e^{2\lambda} \frac{m}{r\gamma} \geq \frac{3}{2\gamma}.
\end{equation}
Thus on $\left[r_2,\frac{16r_2}{15}\right]$ we have
\begin{equation}
\frac{2x^2}{3(1-\gamma x)} - \frac{x}{\gamma} \geq 0.
\end{equation}
Inserting this into (\ref{rxp2}) we obtain on this interval
\begin{equation} \label{rxp3}
rx' \geq \frac{x^2}{3(1-2\gamma x)} -x \geq \frac 4 3 x^2 -x.
\end{equation}
By virtue of Lemma \ref{lem_est_gamma} we have for $r\in[R_0,2R_0]$ the estimate $\gamma \leq  C_\gamma r^{\frac{2}{q+1}}$. So we get
\begin{equation} \label{est_x_r2}
x(r_2) = \frac{m(r_2)}{r_2} \frac{1}{\gamma(r_2)} \geq \frac{2}{5} \frac{C}{r_2^{\frac{2}{q+1}}}.
\end{equation}
Thus $x(r_2) \to \infty$ as $R_0 \to 0$, and we take $R_0$ sufficiently small so that 
\begin{equation}
\frac{x(r_2)}{x(r_2)-\frac 3 4} \leq \frac{16}{15}.
\end{equation}
 In particular we have $x(r_2) \geq 12 > \frac 3 4$. \par
We solve the differential inequality (\ref{rxp3}) with the method of separation of variables and obtain
\begin{equation}
x(r) \geq \frac 3 4 \frac{1}{1-\frac{r\left(\frac{3}{4} x(r_2)-1\right)}{\frac 4 3 r_2 x(r_2)}}.
\end{equation}
We observe that the right hand side diverges if $r$ approaches 
\begin{equation}
\frac{\frac 4 3 r_2 x(r_2)}{\frac 4 3 x(r_2) -1} \leq r_2\frac{16}{15}.
\end{equation} Note that this upper bound is already smaller than $2R_0$ for $R_0$ small. The estimate (\ref{est_x_r2}) yields
\begin{equation}
\frac{x(r_2)}{x(r_2)-\frac 3 4} \leq \frac{\frac 2 5 \frac{C}{r_2^{2/(q+1)}}}{\frac 2 5  \frac{C}{r_2^{2/(q+1)}} - \frac{3}{4}} = \frac{1}{1-\frac{15}{8}\frac{r_2^{2/(q+1)}}{C}}.
\end{equation}
Lemma \ref{lem_gammas} provides 
\begin{equation}
r_2 \leq R_0 + \sigma_*+\kappa R_0^{\frac{q+2}{q+1}}.
\end{equation} 
Hence it is possible to choose a constant $B_1$ such that
\begin{equation}
\frac{x(r_2)}{x(r_2)-\frac 3 4} \leq \frac{1}{1-B_1R_0^{\frac{2}{q+1}}} \to 1, \quad \mathrm{as}\;R_0\to 0.
\end{equation}
Finally, since $\sigma_* \leq \kappa R_0^{\frac{q+2}{q+1}}$, if $R_0$ is sufficiently small, it is clear that there is a positive constant $B_2$ such that 
\begin{equation}
r_2 \leq R_0 + B_2R_0^{\frac{q+2}{q+1}}.
\end{equation} 
In total
\begin{equation}
R_1 \leq \frac{R_0+B_2 R_0^{\frac{q+2}{q+1}}}{1-B_1R_0^{\frac{2}{q+1}}},
\end{equation}
as desired. \par
\end{proof}

\subsection{Proof of the main theorem}

With those lemmas and propositions at hand we can prove the main theorem.
\begin{proof}[Proof of the main theorem]
Suppose $y$, $\varrho$, and $p$ belong to a quasi shell solution with sufficiently small inner radius $R_0$. Proposition \ref{main_theorem} yields the existence of a radius, say $r_1 \leq R_1$, such that $\varrho(r_1)=p(r_1)=0$.  At this point the matter distribution function $f$ can be continued by constant zero and the metric component $y = \ln(E_0)-\mu$ can be continued with the Schwarzschild component
\begin{equation} \label{yss}
y_S(r) = \ln(E_0) -\frac 1 2 \ln\left[1-\frac{2M}{r}\right],
\end{equation}
where the mass parameter $M$ is chosen to be
\begin{equation}
M = m(R_1) = 4\pi \int_{R_0}^{R_1} s^2\varrho(s) \mathrm ds.
\end{equation}
A few words on the regularity of the functions $y$, $\varrho$, and $p$ are in order. By the same reasoning as in Lemma 2.1 in \cite{rr92_static} one obtains that $\varrho$ can be differentiated $\lfloor \ell + \frac 3 2\rfloor$ times with respect to $r$ and the derivative will be continuous. For $p$ one finds that it is $\lfloor\ell + \frac 5 2\rfloor$ times continuously differentiable. The derivative of $y$ is given by (\ref{eq_yp}). One observes that the derivative $y^{(n)}$ is given by an expression containing the functions $p,\dots,p^{(n-1)}$, $m$, and, $\varrho,\dots,\varrho^{(n-2)}$. So by the upper observations one concludes that $y$ can be continuously differentiated $\lfloor\ell + \frac 7 2\rfloor$ times.\par
All existing derivatives of the matter quantities contain integrals over the same domain as in (\ref{mat21}) and (\ref{mat22}). Thus they vanish if and only if $\gamma(r) \leq 0$. So since $\gamma(r_1)=0$ all existing derivatives of the matter quantities $\varrho$ and $p$ are zero and all existing derivatives of $y$ equal the corresponding derivatives of the Schwarzschild component $y_S$. The function $y$ gives rise to a full solution to the Einstein-Vlasov system as for example discussed in \cite{rr92} for the massive case. This ends the proof of the main theorem.
\end{proof}

\section{Criteria for shell formation} \label{sec_numerics}

In this section we present numerical support for Conjecture \ref{our_con}. To this end we construct spherically symmetric solutions of the Einstein-Vlasov system numerically. A typical energy density containing a shell of massless matter can be seen in Figure \ref{shell1}. The peak is enclosed by two vacuum regions. Note that the energy density is not zero for large radii even though it is very small. However, not all solutions show matter regions that are separated by vacuum, as for example the solution depicted in Figure \ref{shell2}.\par
\begin{figure}[h]
\begin{minipage}{0.49\textwidth}
\begin{center}
\includegraphics[height=4.5cm]{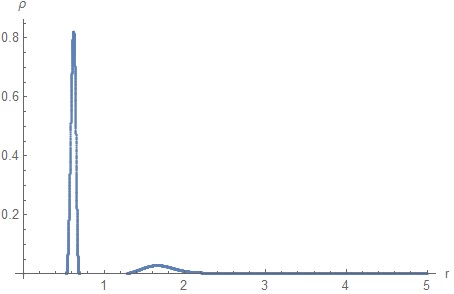}
\caption{A typical shell configuration. $k = 0$, $\ell = \frac 1 2$, $L_0 = 0.8$, $y_0=0.51$  \label{shell1}}
\end{center}
\end{minipage}
\hfill
\begin{minipage}{0.49\textwidth}
\begin{center}
\includegraphics[height=4.5cm]{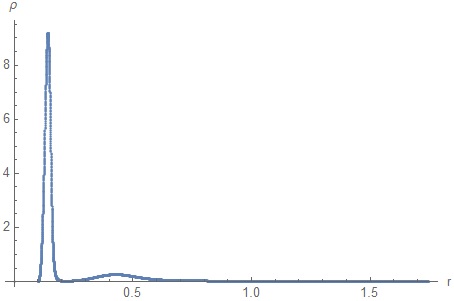}
\caption{The matter peaks are not separated by vacuum. $k = 0$, $\ell = \frac 3 2$, $L_0 = 0.2$, $y_0=0.22$ \label{shell2}}
\end{center}
\end{minipage}
\end{figure}
In Conjecture \ref{our_con} the quantity
\begin{equation}
\Gamma = \sup_{r\in[0,\infty)} \frac{2m(r)}{r}
\end{equation}
was introduced. 
Our numerical calculations indicate that $\Gamma \geq 0.8$ is a necessary condition for that a solution with separated matter regions as in Figure \ref{shell1} occurs. There are four parameters that determine a solution of the Einstein-Vlasov system if the matter distribution function is chosen to be of the form (\ref{ansf}). These parameters are the initial value $y_0$ for the metric component and the constants $k$, $\ell$, and $L_0$ in the ansatz (\ref{ansf}). We observe that this parameter space is separated by a hypersurface into a region where matter shells occur and a region where the support of the matter quantities is connected. Figure \ref{table} shows a part of the plane in this space that is spanned by $\ell$ and $y_0$. Shells only occur if $\Gamma \geq 0.8$. \par
\begin{figure}[h]
\begin{center}
\begin{tabular}{r||ccccccc}
\diagbox{$y_0$:}{$\ell$:} & -0.3 & -0.2 & -0.1 & 0 & 0.1 & 0.2 & 0.3 \\
\hline \hline
\multirow{2}{*}{0.43} & $\bullet$ & $\bullet$ & $\bullet$ & $\bullet$ & $\times$ & $\times$ & $\times$\\
& \scriptsize{0.86} & \scriptsize{0.85} & \scriptsize{0.83} & \scriptsize{0.82} & \scriptsize{0.80} & \scriptsize{0.78} & \scriptsize{0.77}\\
\hline
\multirow{2}{*}{0.36} & $\bullet$ & $\bullet$ & $\bullet$ & $\bullet$ & $\times$ & $\times$ & $\times$\\
& \scriptsize{0.86} & \scriptsize{0.84} & \scriptsize{0.83} & \scriptsize{0.81} & \scriptsize{0.8} & \scriptsize{0.77} & \scriptsize{0.76}\\
\hline
\multirow{2}{*}{0.29} & $\bullet$ & $\bullet$ & $\bullet$ & $\bullet$ & $\times$ & $\times$ & $\times$\\
& \scriptsize{0.86} & \scriptsize{0.84} & \scriptsize{0.82} & \scriptsize{0.80} & \scriptsize{0.79} & \scriptsize{0.77} & \scriptsize{0.75}\\
\hline
\multirow{2}{*}{0.22} & $\bullet$ & $\bullet$ & $\bullet$ & $\bullet$ & $\times$ & $\times$ & $\times$\\
& \scriptsize{0.85} & \scriptsize{0.84} & \scriptsize{0.82} & \scriptsize{0.80}$^*$ & \scriptsize{0.78} & \scriptsize{0.76} & \scriptsize{0.75}\\
\hline
\multirow{2}{*}{0.16} & $\bullet$ & $\bullet$ & $\bullet$ & $\times$ & $\times$ & $\times$ & $\times$\\
& \scriptsize{0.85} & \scriptsize{0.83} & \scriptsize{0.81} & \scriptsize{0.80} & \scriptsize{0.77} & \scriptsize{0.76} & \scriptsize{0.74}\\
\hline
\multirow{2}{*}{0.11} & $\bullet$ & $\bullet$ & $\bullet$ & $\times$ & $\times$ & $\times$ & $\times$\\
& \scriptsize{0.85} & \scriptsize{0.83} & \scriptsize{0.81} & \scriptsize{0.79} & \scriptsize{0.77} & \scriptsize{0.75} & \scriptsize{0.73}
\end{tabular}
\caption{Behavior for different parameter choices. In (\ref{ansf}) we fix $k=0$, $L_0 = 0.05$ and vary $\ell$ and $y_0$. If a shell is obtained we set a $\bullet$, if not a $\times$. The small numbers represent the quantity $\Gamma=\sup_r \frac{2m}{r}$ and the shell solution with minimal $\Gamma$ is indicated with a $^*$.\label{table}
}
\end{center}
\end{figure}
Also the influence of increasing $k$ or $L_0$ has been investigated. One observes in both cases that for each value of $\ell$ a shell as in Figure \ref{shell1} only forms for larger values of $y_0$. The qualitative behavior however -including the observation that all shells have $\Gamma \geq 0.8$- does not change. This behavior was expected since the proof of Theorem \ref{new_main_theorem} works for general $k$ and $L_0$.

H{\aa}kan Andr\'{e}asson\\
          Mathematical Sciences\\
          Chalmers University of Technology\\
          University of Gothenburg\\
          S-41296 G\"oteborg, Sweden\\
          email: hand@chalmers.se\\
          \ \\
          David Fajman\\
          Arbeitsbereich Gravitationsphysik\\
          Universit\"at Wien\\
	     A-1090 Wien, Austria\\
          email: David.Fajman@univie.ac.at\\
	    \ \\
          Maximilian Thaller\\
          Mathematical Sciences\\
          Chalmers University of Technology\\
          University of Gothenburg\\
          S-41296 G\"oteborg, Sweden\\
          email: maxtha@chalmers.se
\end{document}